\documentclass[conference,letterpaper]{IEEEtran}

\usepackage[a4paper, total={210mm, 297mm}, top=0.75in, bottom=1in, left = 0.62in, right=0.62in]{geometry}
\usepackage{graphicx}
\usepackage{color}
\usepackage{pgf, tikz, pgfplots}
\usepackage{siunitx}
\pgfplotsset{compat=1.17} 
\usepackage{amssymb} 
\usepackage{bm}
\usepackage{acronym}
\usepackage{tikz}
\usepackage{mathrsfs}
\usepackage{enumerate}
\usepackage{todonotes}
\usepackage[utf8]{inputenc} 
\usepackage[T1]{fontenc}
\usepackage{url}
\usepackage{ifthen}
\usepackage{cite}
\usepackage[cmex10]{amsmath} %
\usepackage{braket}
\usepackage{stmaryrd}

\usetikzlibrary{arrows,calc,fit,matrix,positioning,spy}
\tikzset{myblock/.style={rectangle, draw, thin, minimum width=0.6cm, minimum height=0.6cm},font=\footnotesize,align=center}%

\usepackage{cite}
\usepackage{amsmath,amssymb,amsfonts}

\usepackage{dsfont}
\usepackage{algorithm}
\usepackage[noend]{algorithmic}
\usepackage{graphicx}
\usepackage{textcomp}
\usepackage{xcolor}
\def\BibTeX{{\rm B\kern-.05em{\sc i\kern-.025em b}\kern-.08em
    T\kern-.1667em\lower.7ex\hbox{E}\kern-.125emX}}

    \usepackage{amsfonts}     
\usepackage{graphicx}
\usepackage{color}
\usepackage{pgf, tikz, pgfplots}
\usepackage{siunitx}
\pgfplotsset{compat=1.17} 
\usepackage{amssymb} 
\usepackage{bm}
\usepackage{bbm}
\usepackage{amsthm}
\usepackage{acronym}
\usepackage{tikz}
\usepackage{mathrsfs}
\usepackage{enumerate}
\usetikzlibrary{arrows,calc,fit,matrix,positioning,shapes,shadows,trees,mindmap,tikzmark,arrows.meta,angles,quotes,babel,spy}

\usepackage{booktabs}

\definecolor{KITblack}{rgb}{0.25,0.25,0.25}
\definecolor{KITbrown}{rgb}{0.65,0.51,0.18}

\definecolor{kit-green100}{rgb}{0,.59,.51}
\definecolor{kit-green100}{rgb}{0,.59,.51}
\definecolor{kit-green70}{rgb}{.3,.71,.65}
\definecolor{kit-green50}{rgb}{.50,.79,.75}
\definecolor{kit-green30}{rgb}{.69,.87,.85}
\definecolor{kit-green15}{rgb}{.85,.93,.93}
\definecolor{KITgreen}{rgb}{0,.59,.51}

\definecolor{KITpalegreen}{RGB}{130,190,60}
\colorlet{kit-maigreen100}{KITpalegreen}
\colorlet{kit-maigreen70}{KITpalegreen!70}
\colorlet{kit-maigreen50}{KITpalegreen!50}
\colorlet{kit-maigreen30}{KITpalegreen!30}
\colorlet{kit-maigreen15}{KITpalegreen!15}

\definecolor{KITblue}{rgb}{.27,.39,.66}
\definecolor{kit-blue100}{rgb}{.27,.39,.67}
\definecolor{kit-blue70}{rgb}{.49,.57,.76}
\definecolor{kit-blue50}{rgb}{.64,.69,.83}
\definecolor{kit-blue30}{rgb}{.78,.82,.9}
\definecolor{kit-blue15}{rgb}{.89,.91,.95}

\definecolor{KITyellow}{rgb}{.98,.89,0}
\definecolor{kit-yellow100}{cmyk}{0,.05,1,0}
\definecolor{kit-yellow70}{cmyk}{0,.035,.7,0}
\definecolor{kit-yellow50}{cmyk}{0,.025,.5,0}
\definecolor{kit-yellow30}{cmyk}{0,.015,.3,0}
\definecolor{kit-yellow15}{cmyk}{0,.0075,.15,0}

\definecolor{KITorange}{rgb}{.87,.60,.10}
\definecolor{kit-orange100}{cmyk}{0,.45,1,0}
\definecolor{kit-orange70}{cmyk}{0,.315,.7,0}
\definecolor{kit-orange50}{cmyk}{0,.225,.5,0}
\definecolor{kit-orange30}{cmyk}{0,.135,.3,0}
\definecolor{kit-orange15}{cmyk}{0,.0675,.15,0}

\definecolor{KITred}{rgb}{.63,.13,.13}
\definecolor{kit-red100}{cmyk}{.25,1,1,0}
\definecolor{kit-red70}{cmyk}{.175,.7,.7,0}
\definecolor{kit-red50}{cmyk}{.125,.5,.5,0}
\definecolor{kit-red30}{cmyk}{.075,.3,.3,0}
\definecolor{kit-red15}{cmyk}{.0375,.15,.15,0}

\definecolor{KITpurple}{RGB}{160,0,120}
\colorlet{kit-purple100}{KITpurple}
\colorlet{kit-purple70}{KITpurple!70}
\colorlet{kit-purple50}{KITpurple!50}
\colorlet{kit-purple30}{KITpurple!30}
\colorlet{kit-purple15}{KITpurple!15}

\definecolor{KITcyanblue}{RGB}{80,170,230}
\colorlet{kit-cyanblue100}{KITcyanblue}
\colorlet{kit-cyanblue70}{KITcyanblue!70}
\colorlet{kit-cyanblue50}{KITcyanblue!50}
\colorlet{kit-cyanblue30}{KITcyanblue!30}
\colorlet{kit-cyanblue15}{KITcyanblue!15}

\newcommand\blfootnote[1]{%
    \begingroup
    \renewcommand\thefootnote{}\footnote{#1}%
    \addtocounter{footnote}{-1}%
    \endgroup
}

\usepackage{amsmath,amssymb,amsfonts}
\usepackage{algorithmic}
\usepackage{graphicx}
\usepackage{textcomp}
\usepackage{xcolor}
\usepackage[normalem]{ulem}

\usepackage{stmaryrd}
\usepackage{physics}
\newtheorem{definition}{Definition}

\newtheorem{lemma}{Lemma}
\newtheorem{theorem}{Theorem}

\newtheorem{remark}{Remark}
\usepackage[capitalize]{cleveref}
\Crefname{equation}{Eq.}{Eqs.}
\Crefname{definition}{Def.}{Defs.}
\Crefname{theorem}{Thm.}{Thms.}
\Crefname{lemma}{Lem.}{Lems.}
\Crefname{proposition}{Prop.}{Props.}
\Crefname{construction}{Construction}{Constructions}
\Crefname{corollary}{Corollary}{Corollaries}
\Crefname{remark}{Remark}{Remarks}
\Crefname{example}{Example}{Examples}

\usepackage[capitalize]{cleveref}
\Crefname{equation}{Eq.}{Eqs.}
\Crefname{definition}{Def.}{Defs.}
\Crefname{theorem}{Thm.}{Thms.}
\Crefname{lemma}{Lem.}{Lems.}
\Crefname{construction}{Construction}{Constructions}
\newif\ifcomment

\newif\ifhighlightShortLongDiffs
\newif\ifversionShort

\newcommand{\defeq}{:=}

\newcommand{\eqdef}{=:}

\newcommand{\myspan}[1]{\ensuremath{\left\langle #1\right\rangle}}

\newcommand{\imunit}{\iota}

\newcommand{\wt}{\ensuremath{\mathrm{wt}}}

\def\rz{\ifmmode{\mathds{R}}%
    \else{\hbox{$\mathds{R}$}}\fi} 
\def\nz{\ifmmode{\mathbb{N}}%
    \else{\hbox{$\mathds{N}$}}\fi} 
\def\gz{\ifmmode{\mathds{Z}}%
   \else{\hbox{$\mathds{Z}$}}\fi} 
\def\cz{\ifmmode{\mathds{C}}
    \else{\hbox{$\mathds{C}$}}\fi}%
\def\qz{\ifmmode{\mathds{Q}}%
    \else{\hbox{$\mathds{Q}$}}\fi}%
\def\K{\ifmmode{\mathds{K}}%
    \else{\hbox{$\mathds{K}$}}\fi}%

\newcommand{\0}{\boldsymbol{0}}

\newcommand{\be}{\boldsymbol{e}}
\newcommand{\bg}{\boldsymbol{g}}

\newcommand{\bt}{\boldsymbol{t}}

\newcommand{\bx}{\boldsymbol{x}}

\newcommand{\bz}{\boldsymbol{z}}

\newcommand{\bH}{\boldsymbol{H}}

\newcommand{\bL}{\boldsymbol{L}}
\newcommand{\bM}{\boldsymbol{M}}

\newcommand{\bP}{\boldsymbol{P}}

\newcommand{\bS}{\boldsymbol{S}}

\newcommand{\F}{\ensuremath{\mathbb{F}}}

\newcommand{\cB}{\ensuremath{\mathcal{B}}}

\newcommand{\cD}{\ensuremath{\mathcal{D}}}

\newcommand{\cG}{\ensuremath{\mathcal{G}}}

\newcommand{\cL}{\ensuremath{\mathcal{L}}}

\newcommand{\cP}{\ensuremath{\mathcal{P}}}

\newcommand{\cS}{\ensuremath{\mathcal{S}}}

\newcommand{\bbC}{\ensuremath{\mathbb{C}}}

\newcommand{\bbE}{\ensuremath{\mathbb{E}}}

\newcommand{\bbS}{\ensuremath{\mathbb{S}}}

\usepackage{physics}

\newcommand{\qcode}[1]{\ensuremath{\left\llbracket#1\right\rrbracket}}

\commenttrue

 \versionShortfalse %

\highlightShortLongDiffsfalse

\newlength{\ferfigheight}      %
\begin{document}
\title{
Affine Subcode Ensemble Decoding for \\ Degeneracy-Aware Quantum Error Correction
}

\author{%
\IEEEauthorblockN{
Leo Wursthorn$^\dagger$,
Jonathan Mandelbaum$^\dagger$,
Sisi Miao$^\dagger$,
Hedongliang Liu$^\dagger$,\\
Holger Jäkel$^\dagger$,
Stergios Koutsioumpas$^\star$,
Laurent Schmalen$^\dagger$
}
\IEEEauthorblockA{
$^\dagger$\,Communications Engineering Lab (CEL), Karlsruhe Institute of Technology (KIT), 76187 Karlsruhe, Germany\\
}
\IEEEauthorblockA{
$^\star$\,School of Informatics, University of Edinburgh, Edinburgh EH8 9AB, United Kingdom\\
Email: \texttt{jonathan.mandelbaum@kit.edu}
}
}%

\maketitle

\begin{abstract}
Quantum low-density parity-check codes are promising candidates for low-overhead fault-tolerant quantum computing, but \emph{degeneracy} is known to impair the convergence of belief-propagation (BP) decoding of these codes.
In this work, we show that appending linearly independent rows to a check matrix of a stabilizer code 
can reduce the search space for a valid degenerate solution.
Motivated by this, we extend the recently proposed affine subcode ensemble decoding technique from the classical to the quantum setting.
Moreover, we employ overcomplete matrices for each decoding path.
Monte-Carlo simulations on toric and generalized bicycle codes demonstrate improved convergence and reduced logical error rate.
\end{abstract}

\blfootnote{
   This work was funded by the Deutsche Forschungsgemeinschaft (DFG, German
Research Foundation) –  Project 567557088 and also received funding from the German Federal Ministry of
Research, Technology and Space (BMFTR) within the project
Open6GHub+ (grant agreement 16KIS2405). Parts of this work have received funding from the European
Research Council (ERC) under the European
Union’s Horizon 2020 research and innovation programme (grant agreement No. 101001899)
}

\section{Introduction}

Quantum devices are noisy, making quantum error correction (QEC)
essential for scalable fault-tolerant computation \cite{Shor1995, Gottesman1997,wilde2013quantum}.
Quantum low-density parity-check (QLDPC) codes are a promising class of QEC codes, as their sparse checks make them attractive for low-overhead architectures
\cite{MacKay2004,TillichZemor2014,Panteleev2021,bravyi_high-threshold_2024}.
The sparsity of classical LDPC codes enables efficient graph-based message-passing decoders, in particular belief propagation (BP) \cite{Gallager1962}, which is attractive for real-time QEC.

A general class of QEC codes that encode logical information into the joint eigenspace of commuting Pauli operators is given by stabilizer codes. In the stabilizer framework, decoding is inherently \emph{degenerate}, which means that multiple Pauli errors may share the same syndrome and the same logical effect. Exploiting degeneracy is essential for quantum decoding~\cite{shor1996quantumerrorcorrectingcodesneed, iyer2015hardness, PoulinChung2008, GuidedDecimation2024, MSLS25, Raveendran2021}. 
However, degeneracy can also hinder the application of BP decoding to QEC 
by causing the decoder to oscillate rather than to converge\cite{mueller2025improvedbeliefpropagationsufficient}.
Additionally, non-convergence is exacerbated since QLDPC codes contain many trapping sets.
Thus, several methods have been
proposed to improve the convergence and accuracy of BP-based decoding for QLDPC codes.
Post-processing approaches use BP soft information to construct a syndrome-consistent estimate~\cite{Panteleev2021,roffe_decoding_2020,hillmann_localized_2025, yin_symbreak_2024}, but can increase worst-case latency, which is problematic for real-time QEC. Other approaches modify BP directly, e.g., through scaled messages~\cite{LK21} or neural BP~\cite{liu_neural_2019,MSLS25}.
Recently, ensemble decoding (ED) has gained attention as a strategy for improving robustness by exploiting diversity across multiple decoding paths, both in the classical~\cite{Geiselhart2022,mandelbaum2024endomorphisms, Hehn2010, MMJS24, MJS25,AED_RMcodes} and quantum settings~\cite{ koutsioumpas_automorphism_2025, koutsioumpas_colour_2025, maan2026decodingcorrelatederrorsquantum, gong2024lowlatencyiterativedecodingqldpc, mueller2025improvedbeliefpropagationsufficient, ye2025beamsearchdecoderquantum}.

In this work, we provide a mathematical characterization of degeneracy by introducing \emph{degeneracy sets}.
Building on those, we show that appending additional $\Delta$ linearly independent rows to a check matrix of a stabilizer code and setting the additional syndrome component to any binary vector of length $\Delta$ can split each degeneracy set into equally-sized disjoint subsets.
We then extend \emph{affine subcode ensemble decoding} (aSCED)~\cite{MBtB26} that has recently been introduced in classical coding theory
to the QEC setting, where it naturally exploits this structure and introduces additional diversity.
Finally, we present simulation results for toric codes~\cite{MWPM2002,Kitaev2003} and generalized bicycle (GB) codes~\cite{Panteleev2021}, demonstrating that aSCED significantly improves the convergence behavior of BP and achieves outstanding decoding performance.

\textbf{Notation:} 
Let $\imunit=\sqrt{-1}$ and \mbox{$[n]:=\{1,2,\dots,n\}$}.
Let $q$ be a prime power and $\F_q$ be a Galois field of size $q$.
Let boldface letters denote vectors and matrices, e.g., $\bm{a}$ and $\bm{A}$, where $a_i$ is the $i$th component of $\bm{a}$ and $A_{j,i}$ is the element in the $j$th row and $i$th column of $\bm A$; $\bm A^\top$ denotes the matrix transpose.

\section{Preliminaries}
Consider the Pauli operators in \mbox{$\cP=\{I, X, Y,Z\}\subset\bbC^{2\times 2}$} that are Hermitian, unitary,
and with the commutation properties 
\mbox{$ZX=-XZ=\imunit Y$}. Together with \mbox{$N\defeq \{\pm 1, \pm \imunit \}$} and the matrix multiplication over $\bbC^{2\times 2}$, they form the single-qubit Pauli group \mbox{$\cG_1\defeq \{cP:c\in N, P\in \cP\}$}.
We extend $\cG_1$ to \mbox{$\mathcal{G}_n:=\{c\bP=cP_1\otimes \cdots \otimes P_n: c\in N, P_i\in\cP\}$}, the $n$-qubit Pauli group, via the Kronecker product $\otimes$.
The weight of a Pauli operator \mbox{$\bP\in \cG_n$} is \mbox{$\mathrm{wt}\left(\bm{P}\right) \defeq \left|\left\{i\in[n]: {P}_i \neq I\right\}\right|$}.

Let $\cS$ be a commutative subgroup of $\cG_n$ of order $n-k$, generated by a set of generators ${\mathcal{B}_\mathcal{S}:=\{\bm{S}_1,\ldots, \bm{S}_{n-k}\}}$.
An $\qcode{n,k}\subset\bbC^{2^n}$ \emph{stabilizer code} is defined based on $\cS$ by
\[
Q(\cS)\defeq \left\{\ket{\psi}\in\bbC^{2^n}: \bm{S}\ket{\psi}=(+1)\ket{\psi},\ \forall \bm{S}\in\mathcal{S}\right\}.
\]
Then, $Q(\cS)$ possesses a set ${\mathcal{B}_\mathcal{L}\defeq\{\bL_i\}_{i=1}^{2k}\subset\cG_n\setminus\mathcal{S}}$ of $2k$ independent \emph{logical operators}~\cite{Gottesman1997} such that ${\bL_i\ket{\psi}\in Q(\cS)}$, ${\forall \ket{\psi}\in Q(\cS)}$, but ${\exists \ket{\psi}\in Q(\cS): \bL_i\ket{\psi}\neq\ket{\psi}}$. Let $\cL\subset \cG_n$ be the subgroup generated by $\cB_{\cL}$.
The distance of the code $Q(\cS)$ is defined as the minimum weight of a logical operator, 
i.e., 
$d\defeq \min_{\bL\in\cL}\wt(\bL)$.
If the minimum distance $d$ is known, we write $\llbracket n,k,d \rrbracket$.

Now, let ${N^{\otimes n}\defeq\{\pm 1, \pm \imunit\}I^{\otimes n}}$ and $\overline{\mathcal{G}}_n:= \cG_n/N^{\otimes n}$ be the $n$-qubit Pauli group modulo the global phase, which can be ignored in QEC \cite{Gottesman1997}.
Note that $\overline{\mathcal{G}}_n$ admits a vector space structure over $\mathbb{F}_2$\cite{Calderbank1998}. 
Using the definitions ${\bm X(\bm{x}):=X^{x_1}\otimes\cdots\otimes X^{x_n}}$ and ${\bm Z(\bm{z}):=Z^{z_1}\otimes\cdots\otimes Z^{z_n}}$, %
we can associate $n$-qubit Pauli operators with 
binary vectors $\bx, \bz\in\F_2^n$. Then, ${\mathcal{R}\defeq \{X(\bm{x})\cdot Z(\bm{z}):\bx, \bz\in\F_2^n\}}$ consists of the representatives of the cosets of $\overline{\mathcal{G}}_n$ such that ${\{N^{\otimes n}\bm{R}:\bm{R}\in\mathcal{R}\}=\overline{\mathcal G}_n}$.
We use the isomorphism~\cite{Calderbank1998} ${\phi:(\overline{\mathcal G}_n,\cdot)\to(\mathbb{F}_2^{2n},+)}$
with \mbox{${\phi(N\bm{R})\defeq\phi(\bm{R})\defeq(\bm{x}|\bm{z)}}$} for \mbox{$\bm{R}=\bm X(\bm{x}) \bm Z(\bm{z})\in\mathcal{R}$}, and denote its inverse by $\phi^{-1}$. Observe that $\phi$ is well-defined on equivalence classes, i.e., it does not depend on the choice of the representative.
Furthermore, for $n=1$, ${\bm X(\bm{x})=X^{x_1}}$ and ${\bm Z(\bm{z})=Z^{z_1}}$. Then $\phi$ maps single-qubit Pauli operators in $\cP$ to vectors in $\F_2^2$.  
We extend $\phi$ to subsets \mbox{$\mathcal{A}\subseteq\overline{\mathcal G}_n$} by \mbox{$\phi(\mathcal{A})\defeq \{\phi(\bm{A}): \bm{A}\in\mathcal{A}\}$}. Additionally, we equip $\mathbb{F}_2^{2n}$ with the \emph{symplectic inner product} %
\[
\langle\bm{u},\bm{v}\rangle\defeq \sum_{i=1}^n u_{x,i}\cdot v_{z,i} +\sum_{i=1}^n v_{x,i}\cdot u_{z,i}\in \mathbb{F}_2
\]
for \mbox{$\bm{u}=(\bm{u}_x|\bm{u}_z)\in \mathbb{F}_2^{2n}$} and \mbox{$\bm{v}=(\bm{v}_x|\bm{v}_z)\in \mathbb{F}_2^{2n}$}. It can be shown that \mbox{$\langle\bm{u},\bm{v}\rangle=0$} iff $\phi^{-1}(\bm{u})$ and $\phi^{-1}(\bm{v})$ commute\cite{Calderbank1996}. %
A stabilizer code can be characterized by a check matrix
\begin{align}\label{eq:stabilizer-check-matrix}
\bm{H}=
\begin{pmatrix}
    (\phi(\bm{S}_{1}))^\mathsf{T}\,\, \ldots\,\,    (\phi(\bm{S}_{n-k}))^\top
\end{pmatrix}^\top\in \mathbb{F}_2^{(n-k)\times 2n}.
\end{align}
A check matrix \mbox{$\bm{H}\in \mathbb{F}_2^{m\times 2n}$} is called \textit{overcomplete} iff \mbox{$m>\mathrm{rank}(\bm{H})=n-k$}.

An important subclass of stabilizer codes are \textit{Calderbank--Shor--Steane} (CSS) codes \cite{Calderbank1996, Steane1996} with check matrix
\begin{flalign*}
    \bm{H}=
    \begin{pmatrix}
    \bm{H}_X&\bm{0}_{m_X\times n} \\
    \bm{0}_{m_Z\times n}&\bm{H}_Z
    \end{pmatrix}
    \in \mathbb{F}_2^{m\times 2n},
\end{flalign*} 
where the component parity-check matrices (PCMs) ${\bm{H}_X\in\mathbb{F}_2^{m_X\times n}}$ and $\bm{H}_Z\in\mathbb{F}_2^{m_Z\times n}$ are PCMs of two classical codes that satisfy the \textit{symplectic criterion} \mbox{$\bm{H}_X\bm{H}_Z^\top=\bm{0}$}.

\section{Degeneracy Sets and Decoding}
\subsection{Degeneracy Sets}
\label{sec:degeneracy-sets}
In this work, we assume a depolarizing channel in which the $i$th physical qubit $\ket{\psi_i}$ is modified independently by a random Pauli operator \mbox{$P_i\in\mathcal{P}$}.
Each error $X$, $Z$ and $Y$ occurs with probability $\tfrac{p}{3}$, where $p$ denotes the physical error probability.

The binary representation of Pauli operators allows for the following mathematical characterization of \emph{degeneracy}.
By construction, ${\phi(\mathcal{B}_\mathcal{S})}$ constitutes a basis of an ${(n-k)}$-dimensional subspace $\mathbb{S}\subseteq\mathbb{F}_2^{2n}$.
Similarly, $\phi(\mathcal{B}_\mathcal{L})$ is the basis of a $2k$-dimensional subspace $\mathbb{L}\subseteq\mathbb{F}_2^{2n}$ and all vectors in $\phi(\mathcal{B}_\mathcal{L})$ and $\phi(\mathcal{B}_\mathcal{S})$ are linearly independent.

Now let $\{\bm{\varepsilon}_1, \dots, \bm{\varepsilon}_{n-k}\}\subset \F_2^{2n}$ be chosen such that it complements $\phi(\mathcal{B}_\mathcal{S}\cup\mathcal{B}_\mathcal{L})$ to a basis of $\mathbb{F}_2^{2n}$. 
Then, this set forms the basis of an $(n-k)$-dimensional subspace ${\mathbb{E}\subseteq\mathbb{F}_2^{2n}}$ and 
$\mathbb{F}_2^{2n}=\mathbb{E}\oplus \mathbb{S}\oplus \mathbb{L}$ is a direct sum of subspaces.  
Thus, the binary representation $\bm{e}\in \mathbb{F}_2^{2n}$ of an arbitrary Pauli error can be uniquely decomposed as
\begin{align}
\label{eq:error-decomposition}
\bm{e}=\bm{\varepsilon}+\bm{\sigma}+\bm{\lambda},\quad \bm{\varepsilon}\in \mathbb{E}, \bm{\sigma}\in \mathbb{S}, \bm{\lambda}\in \mathbb{L}.
\end{align}
We characterize the error components as follows: 
\begin{itemize}
    \item $\bm{\varepsilon}$ or $\phi^{-1}(\bm{\varepsilon})$ is the \emph{detectable} error component since it anti-commutes with at least one stabilizer,
    \item $\bm{\sigma}$ or $\phi^{-1}(\bm{\sigma})$ is the \emph{trivially undetectable} component, 
    \item $\bm{\lambda}$ or $\phi^{-1}(\bm{\lambda})$ is the \emph{non-trivial undetectable} error component, also known as a \emph{logical error} in literature.
\end{itemize}
Based on this decomposition, we define a degeneracy set:
\begin{definition}
    A \emph{degeneracy set} $\mathcal{D}_{\bm{\lambda},\bm{\varepsilon}}$ is defined as a coset of $\mathbb{S}$ in $\mathbb{F}_2^{2n}/\mathbb{S}$ resulting in equivalence classes
indexed by $\bm{\lambda}\in \mathbb{L}$ and $\bm{\varepsilon}\in \mathbb{E}$, i.e.,
\[
\mathcal{D}_{\bm{\lambda},\bm{\varepsilon}}\defeq \bm{\lambda}+ \bm{\varepsilon}+\mathbb{S} \defeq \{\bm{\lambda}+ \bm{\varepsilon}+\bm{\sigma}: \bm{\sigma} \in \mathbb{S}\}
.
\]
\end{definition}
For any $\be_1, \be_2\in \mathcal{D}_{\bm{\lambda},\bm{\varepsilon}}$, there exists 
$\bS\in\cS$ such that $\phi^{-1}(\be_1)=\bS \phi^{-1}(\be_2)$, i.e., a degeneracy set summarizes all errors that only differ by a stabilizer.

\subsection{Quantum Decoding Problem}
For an arbitrary error $\bm{e}=\bm{\varepsilon}+\bm{\sigma}+\bm{\lambda}
\in \mathbb{F}_2^{2n}$, the \emph{syndrome} indicates, for each stabilizer in ${\mathcal{B}_\mathcal{S}}$, whether the Pauli operator $\phi^{-1}(\be)$ commutes or anti-commutes with it~\cite{Roffe2019}. 
The extracted syndrome is denoted by $\bz\in\F_2^{n-k}$, where  ${z_j\defeq\langle\be,\phi(\bm{S}_j)\rangle}$ for $\bS_j\in\mathcal{B}_\mathcal{S}$.
Note that the syndrome is determined by the detectable error component $\bm{\varepsilon}$, since for both other components %
${\myspan{\bm{\lambda}, \phi(\bS)}=\myspan{\bm{\sigma}, \phi(\bS)}=0,\, \forall \bS\in\mathcal{S}}$.

A decoder is solving a \emph{syndrome-decoding problem}: given $\bH$ as in \eqref{eq:stabilizer-check-matrix} and given the extracted syndrome $\bz\in \mathbb{F}_2^{n-k}$, find
an estimated error $\hat{\bm{e}}\in\mathbb{F}_2^{2n}$ such that
\begin{align}
    \label{eq:syndrome_match}
\langle \hat{\bm{e}},\bm H_{j, :} \rangle =\langle \hat{\bm{e}},\phi(\bm{S}_j) \rangle =  z_j \quad \forall j\in[n-k],
\end{align} 
where $\bH_{j, :}$ is the $j$th row of $\bH$.

Now, assume that $\be\in \mathcal{D}_{\bm{\lambda},\bm{\varepsilon}}$.  
We distinguish \emph{degeneracy-aware} decoding and \emph{non-degeneracy-aware} decoding. Degeneracy-aware decoding aims at identifying an error estimate $\hat{\be}$ such that 
$\hat{\bm{e}}\in\mathcal{D}_{\bm{\lambda},\bm{\varepsilon}}$,
i.e., the estimated error $\hat{\bm{e}}$ is \emph{logically equivalent} to the occurred error $\be$.
In contrast, non-degeneracy-aware decoding aims at finding $\hat{\bm{e}}=\bm{e}$.
The decoding result can be characterized into four distinct cases:
\begin{itemize}
    \item Type I success: $\hat{\bm{e}}=\bm{e}$, i.e., the error is estimated exactly.
    \item Type II success: $\hat{\bm{e}}\neq\bm{e}$, but is logically equivalent to $\be$.
    \item Type I failure (flagged): no $\hat{\bm{e}}$ consistent with $\bm{z}$ is found.
    \item Type II failure (unflagged): an $\hat{\bm{e}}$ consistent with $\bm{z}$ is found, but $\hat{\bm{e}}\notin\mathcal{D}_{\bm{\lambda},\bm{\varepsilon}}$.
\end{itemize}
The Type II failure will result in a Pauli correction causing a logical error, as \mbox{$\hat{\bm{e}}\in \mathcal{D}_{\bm{\lambda}',\bm{\varepsilon}}$} with \mbox{$\bm{\lambda}'\neq \bm{\lambda}$} induces \mbox{$\phi^{-1}(\hat{\be})\cdot \phi^{-1}(\be)=\phi^{-1}(\bm{\lambda}'-\bm{\lambda})$} that acts non-trivially on the orginal code state.

\subsection{Belief Propagation Decoding for QEC}

BP decoding is an iterative message passing algorithm operating on the \textit{Tanner graph} induced by the check matrix \mbox{$\bm H\in\mathbb{F}_2^{m\times 2n}$} of the stabilizer code. The Tanner graph is a bipartite graph consisting of two disjoint sets of nodes: 
the variable nodes (VNs) \mbox{$\{\mathsf{v}_i\}_{i\in[n]}$}, associated with the physical qubits, and the check nodes (CNs) \mbox{$\{\mathsf{c}_j\}_{j\in[m]}$}, associated with the rows of $\bm H$, referred to as checks. 
A VN $\mathsf{v}_i$ is connected to a CN $\mathsf{c}_j$ via an edge iff \mbox{$h_{j,i}:=(H_{j,i}\vee H_{j,n+i})=1$}, where $\vee$ denotes the logical OR.
We denote the neighboring sets of VNs and CNs as \mbox{$\mathcal{N}(i):=\{j\in[m]: h_{j,i}=1\}$} and \mbox{$\mathcal{M}(j):=\{i\in[n]: h_{j,i}=1\}$}, respectively.

In this work, quaternary belief propagation (BP4) decoding is used, which implements syndrome-based BP decoding over $\mathbb{F}_4$, providing a unified decoding framework that accounts for all three Pauli errors jointly~\cite{LK20, LK21}. In contrast to binary BP~\cite{PoulinChung2008}, which decodes for $X$- and $Z$-errors separately, BP4 captures the correlations from simultaneous $X$- and $Z$-errors and therefore yields better decoding performance~\cite{MSLS25}.

Specifically, we use log-domain refined BP4 decoding~\cite{LK21}.
Typically, the algorithm is introduced in $\mathbb{F}_4^n$; to keep minimal notation, we describe the algorithm in $\mathbb{F}_2^{2n}$.
Each VN $\mathsf{v}_i$, \mbox{$i\in[n]$}, is assigned a %
log-likelihood ratio (LLR) vector \mbox{$\bm{L}_i =\left(L_i^{(X)},\, L_i^{(Z)},\, L_i^{(Y)}\right)\in\mathbb{R}^3$}, with one component for each Pauli error \mbox{$\zeta\in\mathcal{P}\setminus\{I\}$}. 
Let ${\bm{e}_i\defeq(e_i \mid e_{n+i})}$ denote the binary representation of a Pauli error on the $i$th qubit.
Based on a given physical error probability $p_0$ of the depolarizing channel, the a-priori LLRs are defined as 
\begin{flalign*}
    \tilde{L}_i^{(\zeta)} = \ln\left(\frac{P(\bm e_i=\phi(I))}{P(\bm e_i=\phi(\zeta))}\right) = \ln\left(\frac{1 - p_0}{\frac{p_0}{3}}\right), \quad\zeta\in\mathcal{P}\setminus\{I\},
\end{flalign*}
and used to initialize the variable-to-check messages $\bm{L}_{i\rightarrow j}$. 
To reduce complexity without performance loss in the following CN update, a \textit{belief quantization operator} \mbox{$\lambda:\mathbb{R}^3\to\mathbb{R}$} maps each LLR vector $\bm{L}_{i\rightarrow j}$ onto a scalar LLR by
\begin{flalign*}
    \lambda(\bm{L}_{i\rightarrow j}) 
    = \ln\left(\frac{P(\langle \bm e_i,\phi(\eta)\rangle = 0)}{P(\langle \bm e_i,\phi(\eta)\rangle = 1)}\right) 
    = \ln%
    \left(\frac{1+\mathrm{e}^{-L_{i\rightarrow j}^{(\eta)}}}{\underset{\zeta\in\mathcal{P}\setminus\{I,\eta\}}{\sum} \!\!\!\!\!\! \mathrm{e}^{-L_{i\rightarrow j}^{(\zeta)}}}%
    \right)\!,
\end{flalign*}
where \mbox{$\eta=\phi^{-1}((H_{j,i}\mid H_{j,n+i}))$}.
For each CN $\mathsf{c}_j$, \mbox{$j\in[m]$}, the outgoing messages are then computed as
\begin{flalign*}
    L_{i\leftarrow j} = (-1)^{z_j} \cdot 2\tanh^{-1}\!\left(\prod_{i'\in\mathcal{M}(j)\backslash\{i\}}\!\!\!\!\!\!\tanh\left(\frac{\lambda(\bm{L}_{i\rightarrow j})}{2}\right)\right)\!.
\end{flalign*}
The outgoing messages of each VN $\mathsf{v}_i$ are computed for all Pauli errors individually as
\begin{flalign*}
    L_{i\rightarrow j}^{(\zeta)} = \tilde{L}_i^{(\zeta)} + \!\!\!\!\!\!\!\!\!\!\!\!\!\!\!\!\!\!
    \sum_{\substack{j' \in \mathcal{N}(i) \backslash \{j\} \\ 
    \langle \phi(\zeta), (H_{j',i}\mid H_{j',n+i}) \rangle = 1}} \!\!\!\!\!\!\!\!\!\!\!\!\!\!\!\!\!\!\!
    L_{i\leftarrow j'}=: L_i^{(\zeta)}\!-  L_{i\leftarrow j},\quad\zeta\in\mathcal{P}\setminus\{I\}.
\end{flalign*}
A hard decision on all a-posteriori LLRs $\bm{L}_i$, \mbox{$i\in[n]$}, yields the estimate $\hat{\bm{e}}$, with \mbox{$\hat{\bm{e}}_i=\phi(I)$} iff \mbox{$L_i^{(\zeta)}>0,\,\forall\zeta\in\mathcal{P}\setminus\{I\}$}, and \mbox{$\hat{\bm e}_i=\phi\left(\arg\min_\zeta L_i^{(\zeta)}\right)$} otherwise. The iterative decoding process is stopped when \eqref{eq:syndrome_match} is satisfied or $I_\text{max}$ iterations are reached.

\section{Splitting of Degeneracy Sets}
\label{sec:splitters}
The core idea of the decoding approach introduced in the subsequent section is to append additional rows to the check matrix \mbox{$\bm{H}\in \mathbb{F}_2^{m\times 2n}$}, which are linearly independent of the rows of $\bm{H}$. 
Note that we cannot measure their syndrome since they act non-trivially on the encoded state, so we need to virtually fix additional syndrome bits corresponding to such rows.
In the following, we show that if such a linearly appended row is not a logical operator, it has a \emph{splitting} effect on all the degeneracy sets $\mathcal{D}_{\bm{\lambda},\bm{\varepsilon}}$, $\bm{\lambda}\in \mathbb{L}$, $\bm{\varepsilon}\in \mathbb{E}$. 
\begin{lemma}\label{lem:t-independent}
    Let 
    $\bt={\bm{\varepsilon}_\mathsf{t}+\bm{\sigma}_\mathsf{t}+\bm{\lambda}_\mathsf{t}}\in\F_2^{2n}$ be a decomposition with $\bm{\varepsilon}_\mathsf{t}\in \mathbb{E}\setminus\{\bm{0}\}$, $\bm{\sigma}_\mathsf{t}\in \mathbb{S}$, and $\bm{\lambda}_\mathsf{t}\in \mathbb{L}$. Then $\bt$ is linearly independent of the rows of $\bH$, which form a basis of $\bbS$. 
\end{lemma}
\begin{IEEEproof}
    It follows from the definitions of the subspaces $\bbS$ and $\bbE$ in \cref{sec:degeneracy-sets} that $\bbS\cap \bbE=\{\0\}$. By choosing $\bm{\varepsilon}_{\sf t}\in\bbE\setminus\{\0\}$, we ensure $\bt\notin \bbS$ and hence $\bt$ is linearly independent of the rows of $\bH$.
\end{IEEEproof}
\begin{theorem}[Splitting the Degeneracy Sets]
\label{thm:t-splitting}
    Let $\bt$ be defined as in Lemma~1, i.e., with $\bm{\varepsilon}_\mathsf{t}\in \mathbb{E}\setminus\{\bm{0}\}$, and let \mbox{$z_{m+1}=g\in \mathbb{F}_2$} be the additional non-measured ($m+1$)th syndrome bit corresponding to $\bt$.
    Then, the additional row $\bm{t}$ \emph{splits} each degeneracy set $\mathcal{D}_{\bm{\lambda},\bm{\varepsilon}}$ , $\bm{\lambda}\in \mathbb{L}$, $\bm{\varepsilon}\in \mathbb{E}$, into two disjoint subsets 
    \begin{align*}
 \mathcal{D}^{(g)}_{\bm{\lambda},\bm{\varepsilon}}&=\{\bm{e}\in\mathcal{D}_{\bm{\lambda},\bm{\varepsilon}} : \langle \bm{e},\bm{t}\rangle=g\}, \, g\in \mathbb{F}_2,
    \end{align*}
    of equal size $|\mathcal{D}^{(0)}_{\bm{\lambda},\bm{\varepsilon}}|=|\mathcal{D}^{(1)}_{\bm{\lambda},\bm{\varepsilon}}|=\frac{|\mathcal{D}_{\bm{\lambda},\bm{\varepsilon}}|}{2}=2^{n-k-1}$.
\end{theorem}
\begin{IEEEproof}
Due to Lemma~\ref{lem:t-independent}, $\bt$ is linearly independent from $\phi(\mathcal{B}_\mathcal{S})$ and the Pauli operator $\phi^{-1}(\bm{t})$ anti-commutes with at least one stabilizer $\bS\in\cS$ since the stabilizer are generated by $\mathcal{B}_\mathcal{S}$. Thus, there exists at least one \mbox{$\bm{\sigma}\in \phi(\mathcal{B}_\mathcal{S})$} such that \mbox{$\langle\bm{t},\bm{\sigma}\rangle=1$}.
Without loss of generality, assume that $\langle\bm{t},\bm{\sigma}_{n-k}\rangle=1$, i.e., $\phi^{-1}(\bm{t})$ anti-commutes with the $(n-k)$th stabilizer $\bm{S}_{n-k}$.
We perform a change of basis for $\bbS$: 
\[
\tilde{\bm{\sigma}}_j:=
\begin{cases}
    \bm{\sigma}_j,&\text{if }\langle\bm{t},\bm{\sigma}_{j}\rangle=0,\\
    \bm{\sigma}_j+\bm{\sigma}_{n-k},&\text{if }\langle\bm{t},\bm{\sigma}_{j}\rangle=1,\\
\end{cases} \ j\in[n-k-1],
\]
and \mbox{$\tilde{\bm{\sigma}}_{n-k}={\bm{\sigma}}_{n-k}$}. 
Then \mbox{$\{\tilde{\bm{\sigma}}_j\}_{j=1}^{n-k}$} is another basis of $\mathbb{S}$ such that 
\mbox{$\langle \tilde{\bm{\sigma}}_j,\bm{t}\rangle=0$}, $j \in [n-k-1]$ and \mbox{$\langle \tilde{\bm{\sigma}}_{n-k},\bm{t}\rangle=1$}.
Thus, the set $\{\tilde{\bm{\sigma}}_j\}_{j=1}^{n-k-1}$ forms a basis of a \mbox{$(\mathrm{dim}(\mathbb{S})-1)$}-dimensional subspace $\tilde{\mathbb{S}}\subset\mathbb{S}$ consisting of $2^{n-k-1}$ binary vectors ${\tilde{\bm{\sigma}}\in \tilde{\mathbb{S}}}$  such that  $\langle\tilde{\bm{\sigma}},\bm{t}\rangle=0$.
Furthermore, the coset ${\tilde{\bm{\sigma}}_{n-k}+\tilde{\mathbb{S}}}$ consists of $2^{n-k-1}$ binary vectors ${\tilde{\bm{\sigma}}\in \tilde{\bm{\sigma}}_{n-k}+\tilde{\mathbb{S}}}$  such that  $\langle\tilde{\bm{\sigma}},\bm{t}\rangle=1$.

Consider an arbitrary $\mathcal{D}_{\bm{\lambda},\bm{\varepsilon}}$
and set \mbox{$\beta:=\langle\bm{\lambda}+\bm{\varepsilon},\bm{t}\rangle\in\mathbb{F}_2$}.
Let \mbox{$\mathcal{D}^{(\beta)}_{\bm{\lambda},\bm{\varepsilon}}:=\bm{\lambda}+\bm{\varepsilon}+\tilde{\mathbb{S}}$} and \mbox{$\mathcal{D}^{(\beta+1)}_{\bm{\lambda},\bm{\varepsilon}}:=\bm{\lambda}+\bm{\varepsilon}+(\tilde{\bm{\sigma}}_{n-k}+\tilde{\mathbb{S}})$} with superscript operation conducted in $\mathbb{F}_2$.
Thus, by construction ${|\mathcal{D}^{(\beta)}_{\bm{\lambda},\bm{\varepsilon}}|=|\mathcal{D}^{(\beta+1)}_{\bm{\lambda},\bm{\varepsilon}}|=\frac{|\mathcal{D}_{\bm{\lambda},\bm{\varepsilon}}|}{2}=2^{n-k-1}}$.

Note that for any $\bm{e}\in \mathcal{D}^{(\beta)}_{\bm{\lambda},\bm{\varepsilon}}$ it follows that ${\bm{e}=\bm{\lambda}+\bm{\varepsilon}+\tilde{\bm{\sigma}}}$ for some $\tilde{\bm{\sigma}}\in \tilde{\mathbb{S}}$, and, therefore,
\[
\myspan{\bm{e},\bm{t}}%
=\myspan{\bm{\lambda}+\bm{\varepsilon},\bm{t}}+\myspan{\tilde{\bm{\sigma}},\bm{t}}=\beta+0=\beta.
\]
Similarly, for any $\bm{e}\in \mathcal{D}^{(\beta+1)}_{\bm{\lambda},\bm{\varepsilon}}$ it follows that $\myspan{\bm{e},\bm{t}}=\beta+1$.
Choosing $g=\myspan{\bm{e},\bm{t}}$ concludes the proof.
\end{IEEEproof}
\begin{remark}
Since such rows split the degeneracy sets, we call them \emph{splitters}.    
\end{remark}

\begin{remark}\label{remark:splitters}
     By appending $\Delta$ linearly independent splitters 
     with 
     linearly independent components \mbox{$\bm{\varepsilon}\in \mathbb{E}\setminus\{\bm{0}\}$} to the check matrix, each degeneracy set $\mathcal{D}_{\bm{\lambda},\bm{\varepsilon}}$ is split into $2^\Delta$ subsets of size $2^{n-k-\Delta}$, indexed by the respective preset syndrome component $\bm{g}\in \mathbb{F}_2^{\Delta}$, e.g., $\mathcal{D}^{(\bm{g})}_{\bm{\lambda},\bm{\varepsilon}}$.
\end{remark}

\section{Affine Subcode Ensemble Decoding}

Ensemble decoding employs $K$ constituent decoding paths in parallel to produce a set of diverse estimates $\hat{\be}_p$, \mbox{$p\in[K]$}, based on the extracted syndrome $\bm z$. The final estimate $\hat{\be}$ is selected from a candidate set \mbox{$\mathcal{V} \defeq \{\hat{\bm{e}}_p : \hat{\bm{e}}_p \text{ satisfies } \eqref{eq:syndrome_match}\}$} according to the minimum-weight selection \cite{MMJS24}
\begin{flalign*}
 \hat{\bm{e}}=\underset{\hat{\bm{e}}_p\in\mathcal{V}}{\arg \min}\: \mathrm{wt}(\phi^{-1}(\hat{\bm{e}}_p)),
\end{flalign*}
where \mbox{$\left| \mathcal{V} \right|\!=0$} results in a Type I failure. 

\emph{Affine subcode ensemble decoding} (aSCED)~\cite{MBtB26} leverages \mbox{$K=L\cdot 2^\Delta\in \mathbb{N}$} decoding paths, where $L\in \mathbb{N}$ is the number of \emph{batches}, and $2^\Delta\in \mathbb{N}$ is the number of paths in each batch.
\cref{fig:aSCED_batch} demonstrates the structure of one aSCED batch.
Each batch $\ell\in[L]$ decodes on a distinct extended check matrix 
\begin{flalign*}
    \bm H^{(\ell)} = 
    \begin{pmatrix}
        \bm H_X & \bm 0_{m_X\times n} \\
        \bm A_{X}^{(\ell)} & \bm 0_{\tfrac{\Delta}{2}\times n} \\
        \bm 0_{m_Z\times n} & \bm H_Z \\
        \bm 0_{\tfrac{\Delta}{2}\times n} & \bm A_{Z}^{(\ell)}
\end{pmatrix}\in\F_2^{(m+\Delta) \times 2n},
\end{flalign*}
where \mbox{$\bm A_{X}^{(\ell)},\bm A_{Z}^{(\ell)}\in\mathbb{F}_2^{\frac{\Delta}{2} \times n}$} each consist of $\tfrac{\Delta}{2}$ rows 
such that \mbox{$\text{rank}(\bm H^{(\ell)})=\text{rank}(\bm H)+\Delta$}. 
The rows of $\bm A_{\zeta}^{(\ell)}$, \mbox{$\zeta\in\{X,Z\}$} are generated using~\cite[Algorithm 1]{MJS25} to avoid introducing additional cycles of length $4$ in the Tanner graph representation of the respective component PCMs.

\begin{figure}[t]
    \centering
    \includegraphics[width=0.85\columnwidth]{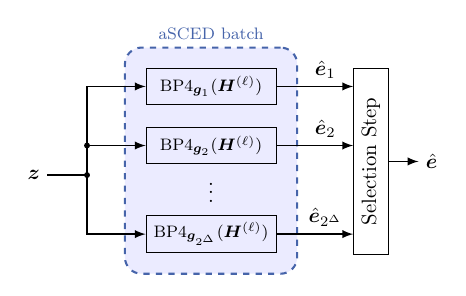}
    \vspace{-0.4cm}
    \caption{Block diagram of an aSCED batch with $2^\Delta$ distinct paths decoding with the same extended check matrix $\bm H^{(\ell)}$. Each path uses a different syndrome configuration $\bm g_\delta\in\mathbb{F}_2^{\Delta}$, $\delta\in[2^\Delta]$, corresponding to the $\Delta$ splitters.
    }
    \label{fig:aSCED_batch}
\end{figure}

For $\delta\in[2^\Delta]$, the $\delta$th path in the $\ell$th batch aims at solving the decoding problem given $\bH^{(\ell)}$ and the extended syndrome 
\begin{flalign}
\label{eq:path-syndrome}
    \bm z_\delta= \left( \bm z_X,\ \bm g_{X,\delta},\ \bm z_Z,\ \bm g_{Z,\delta}\right)\in\mathbb{F}_2^{m+\Delta},
\end{flalign}
where \mbox{$(\bm g_{X,\delta},\bm g_{Z,\delta})\eqdef \bg_\delta\in\mathbb{F}_2^\Delta$} denotes the binary representation of $\delta-1$. 
Hence,  the paths in one batch cover all possible syndrome configurations $\bg_\delta$ corresponding to the additional rows in $\bm A_{\zeta}^{(\ell)}$, \mbox{$\zeta\in\{X,Z\}$}.

According to Lemma~1 and Remark~\ref{remark:splitters}, if the rows do not correspond to logical operators, the $\Delta$ additional rows of $\bm{H}^{(\ell)}$
act as splitters that partition each degeneracy set $\cD_{\bm{\lambda},\bm{\varepsilon}}$ into $2^\Delta$ disjoint subsets $\cD^{(\bm{g}_\delta)}_{\bm{\lambda},\bm{\varepsilon}}$ of size $2^{n-k-\Delta}$. 
Thus, the $\delta$th decoding path in the $\ell$th batch aims at identifying an error estimate $\hat{\bm{e}}$ in the \emph{local degeneracy set} $\cD^{(\bm{g}_\delta)}_{\bm{\lambda},\bm{\varepsilon}}$, while the $\ell$th batch maintains degeneracy-aware decoding on the full degeneracy set
${\cD_{\bm{\lambda},\bm{\varepsilon}} = \bigcup_{\delta=1}^{2^\Delta}  \cD^{(\bm{g}_\delta)}_{\bm{\lambda},\bm{\varepsilon}}}$.

Furthermore, \emph{overcomplete check matrices} have been shown to enhance the performance of BP decoding and to improve the convergence to valid error estimates \cite{Overcomplete2006, MSLS25}.

Linearly independent rows used for aSCED can provide benefits for constructing overcomplete check matrices~\cite{MBtB26}.
First, they enlarge the pool of redundant checks obtainable through linear combinations of rows in the vertically stacked matrix $[\bm H_\zeta\,;\bm A_{\zeta}^{(\ell)}]$, \mbox{$\zeta\in\{X,Z\}$} and enable more distinct overcomplete check matrices across the ensemble~\cite[Sec.VI.B]{MBtB26}. 
Second, choosing a low weight for those additional rows allows for sparser overcomplete check matrices with fewer short cycles in the Tanner graph~\cite[Sec.V.C]{MBtB26}. %

Note that the syndrome corresponding to an overcomplete check matrix in the $\delta$th path is computed from the extended syndrome in \eqref{eq:path-syndrome} by \mbox{$\bm{z}_{\delta,\mathrm{oc}}^\top = \bm{M}\bm{z}_\delta^\top\in\F_2^{m_{\rm oc}}$}. Hereby, \mbox{$\bM\in\F_2^{m_{\rm oc}\times (m+\Delta)}$} is a full-rank matrix that is used for adding linearly dependent rows to the extended check matrix in order to form the overcomplete matrix via \mbox{$\bm{H}^{(\ell)}_\mathrm{oc}= \bm{M}\bm{H}^{(\ell)}$}.

\section{Numerical Results}

We evaluate the proposed aSCED design for two families of CSS codes: toric codes \cite{Kitaev2003} and generalized bicycle (GB) codes \cite{Panteleev2021}.
In particular, we focus on the $\llbracket 128,2,8\rrbracket$ toric code with \mbox{$m_\text{oc}\!=\!384$} rows in its overcomplete check matrix, the $\llbracket 46,2,9\rrbracket$ GB code with \mbox{$m_\text{oc}\!=\!800$}, and the higher-rate $\llbracket 126,28,8\rrbracket$ GB code with \mbox{$m_\mathrm{oc}\!=\!126$}.
For constructing sparse overcomplete check matrices, we use the probabilistic search from \cite{Leon1988}.

The numerical results are obtained by Monte-Carlo simulations assuming a depolarizing channel with physical error probability $p$. Note that for initializing BP4, we use some fixed $p_0$, which not necessarily equals $p$, as this approach has empirically been shown to improve performance \cite{MSLS25}.
We assess decoding performance using the
logical error rate (LER), defined as the fraction of trials in which decoding fails to restore the logical state (Type I or II failure).
Simulations are run until at least $400$ logical errors are recorded or $10^6$ trials have been conducted. 

We use \mbox{$\{\text{BP4},\text{OBP4}\}$-aSCED$K$} to denote an ensemble with $K$
constituent decoders of the respective type, where OBP4 refers to decoding with an overcomplete check matrix.
BP4 uses a maximum of \mbox{$I_\text{max}\!=\!25$} decoding iterations, while OBP4 is limited to \mbox{$I_\text{max}\!=\!12$}.
Following \cite{MSLS25}, we fix \mbox{$p_0\!=\!0.49$} and \mbox{$p_0\!=\!0.3$} for OBP4 on the  $\llbracket 128,2,8\rrbracket$ toric and $\llbracket 46,2,9\rrbracket$ GB code, respectively.
For the $\Delta$ splitters in each aSCED batch, we use rows of weight $4$. 

\subsection{Discussion on Ensemble Parameters}

\begin{figure}[t]
    \centering
    \definecolor{kit-green}{RGB}{0, 150, 130}
\definecolor{kit-blue}{RGB}{70, 100, 170}
\definecolor{kit-red}{RGB}{162, 34, 35}
\definecolor{kit-orange}{RGB}{223, 155, 27}
\definecolor{kit-yellow}{RGB}{252, 229, 0}
\definecolor{kit-purple}{RGB}{163, 16, 124}

\begin{tikzpicture}
    \pgfplotsset{grid style={gray}}
	\pgfplotsset{every tick label/.append style={font=\footnotesize}}
  \begin{axis}[
    xshift = 1.5cm,
    xmin = 0.03,   
    xmax = 0.1,
    ymin = 1e-5,   
    ymax = 1,
    ymode = log,
    xtick = {0.04, 0.06, 0.08, 0.10},
    xticklabel style = {/pgf/number format/fixed},
    axis background/.style = {fill=white, mark size=1.5pt},
    xmajorgrids,
    ymajorgrids,
    width = 8.5cm,
    height = \ferfigheight,
    xlabel={$p$},
    ylabel={LER},
    label style={font=\small},
    legend cell align={left},
    legend style={anchor = south east,  at={(1,0)}, fill opacity=1, text opacity = 1,legend columns=1,font=\footnotesize, row sep = 0pt}
  ]

    \addlegendentry{Baseline}
    \addlegendentry{$K\!=\!4,\phantom{0}\, \Delta\!=\!2,\, L\!=\!1$}
    \addlegendentry{$K\!=\!16,\, \Delta\!=\!4,\, L\!=\!1$}
    \addlegendentry{$K\!=\!16,\, \Delta\!=\!2,\, L\!=\!4$}
    \addlegendentry{$K\!=\!64,\, \Delta\!=\!2,\, L\!=\!16$}

    \addplot[
        color = black,
        line width = 1pt,
        mark options = {solid},
        mark size = 2pt,
        dotted
    ] coordinates {(0.01, 1e-5)};

    \addplot[
        color = kit-green, %
        line width = 0.75pt,
        mark options = {solid},
        mark size = 2pt,
        mark = triangle*,
    ] coordinates {(0.01, 1e-5)};

    \addplot[
        color = kit-orange,
        line width = 0.75pt,
        mark options = {solid},
        mark size = 1.5pt,
        mark = square*,
    ] coordinates {(0.01, 1e-5)};

    \addplot[
        color = kit-blue!70,
        line width = 0.75pt,
        mark options = {solid},
        mark size = 3.5pt,
        mark = x,
    ] coordinates {(0.01, 1e-5)};

    \addplot[
        color = kit-blue,
        line width = 0.75pt,
        mark options = {solid},
        mark size = 1.5pt,
        mark= *
    ] coordinates {(0.01, 1e-5)};

    \addplot[
        color = black,
        line width = 1pt,
        mark options = {solid},
        mark size = 2pt,
        dotted
    ]
        coordinates {
            (0.015, 0.017052)
            (0.03, 0.064418)
            (0.045, 0.139311)
            (0.06, 0.231093)
            (0.075, 0.330529)
            (0.09, 0.442626)
            (0.105, 0.565643)
            (0.12, 0.679454)
            (0.135, 0.776119)
            (0.15, 0.855901)
        };

    \addplot[
        color = black,
        line width = 1pt,
        mark options = {solid},
        mark size = 2pt,
        dotted
    ]
        coordinates {
            (0.03, 0.000368)
            (0.04, 0.001387)
            (0.05, 0.004323)
            (0.06, 0.010042)
            (0.07, 0.019989)
            (0.08, 0.036489)
            (0.09, 0.055848)
            (0.1, 0.086389)
        };

    \addplot[
        color = kit-green, %
        line width = 0.75pt,
        mark options = {solid},
        mark size = 2pt,
        mark = triangle*,
    ]
        coordinates {
            (0.02, 0.004154)
            (0.03, 0.014138)
            (0.04, 0.033573)
            (0.05, 0.062784)
            (0.06, 0.105655)
            (0.07, 0.160523)
            (0.08, 0.223261)
            (0.09, 0.304566)
            (0.1, 0.375134)
        };

    \addplot[
        color = kit-orange,
        line width = 0.75pt,
        mark options = {solid},
        mark size = 1.5pt,
        mark = square*,
    ]
        coordinates {
            (0.02, 0.003268)
            (0.03, 0.011132)
            (0.04, 0.025057)
            (0.05, 0.051467)
            (0.06, 0.085594)
            (0.07, 0.135647)
            (0.08, 0.192373)
            (0.09, 0.257163)
            (0.1, 0.332919)
        };

    \addplot[
        color = kit-blue!70,
        line width = 0.75pt,
        mark options = {solid},
        mark size = 3.5pt,
        mark = x,
    ]
        coordinates {
            (0.02, 0.001642)
            (0.03, 0.006652)
            (0.04, 0.017651)
            (0.05, 0.038733)
            (0.06, 0.070538)
            (0.07, 0.11175)
            (0.08, 0.172394)
            (0.09, 0.24173)
            (0.1, 0.315597)
        };

    \addplot[
        color = kit-orange,
        line width = 0.75pt,
        mark options = {solid},
        mark size = 1.5pt,
        mark = square*,
    ]
        coordinates {
            (0.02, 2.8e-05)
            (0.03, 0.000213)
            (0.04, 0.00093)
            (0.05, 0.002749)
            (0.06, 0.006562)
            (0.07, 0.013254)
            (0.08, 0.024987)
            (0.09, 0.041735)
            (0.1, 0.064655)
        };

    \addplot[
        color = kit-blue!70,
        line width = 0.75pt,
        mark options = {solid},
        mark size = 3.5pt,
        mark = x,
    ]
        coordinates {
            (0.03, 0.000118)
            (0.04, 0.000527)
            (0.05, 0.001723)
            (0.06, 0.004368)
            (0.07, 0.010031)
            (0.08, 0.01907)
            (0.09, 0.033258)
            (0.1, 0.05203)
        };

    \addplot[
        color = kit-blue,
        line width = 0.75pt,
        mark options = {solid},
        mark size = 1.5pt,
        mark= *
    ]
        coordinates {
            (0.03, 9.1e-05)
            (0.04, 0.000411)
            (0.05, 0.00132)
            (0.06, 0.003519)
            (0.07, 0.008015)
            (0.08, 0.015414)
            (0.09, 0.026903)
            (0.1, 0.044295)
        };

    \addplot[black, semithick, smooth, no marks, domain=-115:135, samples=120]
      ( { 0.03 + 0.07*(0.215 + 0.03*cos(\x)) },
        { 10^(5*(0.76 + 0.10*sin(\x)) - 5) } );
    \node[black, font=\footnotesize, anchor=south]
      at (axis cs: 0.044, 0.21) {BP4};
 
    \addplot[black, semithick, smooth, no marks, domain=-115:145, samples=120]
      ( { 0.03 + 0.07*(0.215 + 0.025*cos(\x)) },
        { 10^(5*(0.430 + 0.08*sin(\x)) - 5) } );
    \node[black, font=\footnotesize, anchor=north]
      at (axis cs: 0.047, 0.0005) {OBP4};

  \end{axis}
\end{tikzpicture}
    \vspace{-1.15cm}
    \caption{LER of different \mbox{$\{\text{BP4},\text{OBP4}\}$-aSCED$K$} configurations on the $\llbracket 46,2,9\rrbracket$ GB code with stand-alone BP4 (or OBP4) as a baseline. 
    } 
    \label{fig:GB_A3}
\end{figure}

Fig.~\ref{fig:GB_A3} shows the LER of different \mbox{$\{\text{BP4},\text{OBP4}\}$-aSCED$K$} configurations on the $\llbracket 46,2,9\rrbracket$ GB code. All configurations outperform their respective stand-alone BP4 decoder (baseline) across the entire simulated regime. 

Furthermore, we observe that, when increasing the ensemble size $K$, increasing $L$ yields a larger performance gain than increasing $\Delta$, i.e., increasing the number of batches is preferable over increasing the number of paths per batch. 
For instance, for BP4, consider the gap between \mbox{$K\!=\!4, \Delta\!=\!2, L\!=\!1$} and \mbox{$K\!=\!16, \Delta\!=\!4, L\!=\!1$}, compared to the the gap between \mbox{$K\!=\!4, \Delta\!=\!2, L\!=\!1$} and \mbox{$K\!=\!16, \Delta\!=\!2, L\!=\!4$}, where the latter is significantly larger. 
This behavior may be attributed to increased diversity between batches: increasing $L$ allows for more distinct check matrices $\bm H^{(\ell)}$.
Based on this observation, we set \mbox{$\Delta\!=\!2$} for the remainder of the section, unless stated otherwise.
Further simulations show that the performance gain of using an additional path decoding with $\bm H$ vanishes for larger ensemble sizes, e.g., $K\geq16$.

\subsection{Results for Overcomplete Check Matrices}

\begin{figure}[t]
    \centering
    \definecolor{kit-green}{RGB}{0, 150, 130}
\definecolor{kit-blue}{RGB}{70, 100, 170}
\definecolor{kit-red}{RGB}{162, 34, 35}
\definecolor{kit-orange}{RGB}{223, 135, 27}
\definecolor{kit-yellow}{RGB}{252, 229, 0}
\definecolor{kit-purple}{RGB}{163, 16, 124}

\definecolor{aSCED256}{RGB}{245, 0, 131}

\begin{tikzpicture}
    \pgfplotsset{grid style={gray}}
	\pgfplotsset{every tick label/.append style={font=\footnotesize}}
  \begin{axis}[
    xshift = 1.5cm,
    xmin = 0.015,   
    xmax = 0.15,
    ymin = 1e-6,   
    ymax = 1,
    ymode = log,
    xtick = {0.03, 0.06, 0.09, 0.12, 0.15},
    ytick = {1e-6, 1e-5, 1e-4, 1e-3, 1e-2, 1e-1, 1},
    xticklabel style = {/pgf/number format/fixed},
    axis background/.style = {fill=white, mark size=1.5pt},
    xmajorgrids,
    ymajorgrids,
    width = 8.5cm,
    height = \ferfigheight,
    xlabel={$p$},
    ylabel={LER},
    label style={font=\small},
    legend cell align={left},
    legend style={anchor = south east,  at={(1,0)}, fill opacity=1, text opacity = 1,legend columns=1,font=\footnotesize, row sep = 0pt}
  ]

    \addlegendentry{BP4}
    \addlegendentry{OBP4}
    \addlegendentry{MWPM \cite{Pymatching}}
    \addlegendentry{cMWPM \cite{Pymatching}}
    \addlegendentry{OBP4-aSCED64}
    \addlegendentry{OBP4-aSCED256}

    \addplot[
        color = black,
        line width = 1pt,
        mark options = {solid},
        mark size = 2pt,
        dotted
    ]
        coordinates {
            (0.015, 0.017052)
            (0.03, 0.064418)
            (0.045, 0.139311)
            (0.06, 0.231093)
            (0.075, 0.330529)
            (0.09, 0.442626)
            (0.105, 0.565643)
            (0.12, 0.679454)
            (0.135, 0.776119)
            (0.15, 0.855901)
        };

    \addplot[
        color = black,
        line width = 0.5pt,
        mark options = {solid},
        mark size = 2pt,
        dashed,
        dash pattern=on 4pt off 2pt,
    ]
        coordinates {
            (0.015, 0.004073)
            (0.03, 0.00451)
            (0.045, 0.00742)
            (0.06, 0.020493)
            (0.075, 0.049495)
            (0.09, 0.105871)
            (0.105, 0.174827)
            (0.12, 0.280998)
            (0.135, 0.41612)
            (0.15, 0.529521)
        };

    \addplot[
        color = kit-orange,
        line width = 0.75pt,
        mark options = {solid},
        mark size = 2pt,
        mark = o,
        dashed,
        dash pattern=on 4pt off 2pt,
    ]
        coordinates {
            (0.015, 1.8e-05)
            (0.03, 0.000493)
            (0.045, 0.003454)
            (0.06, 0.012725)
            (0.075, 0.035051)
            (0.09, 0.074585)
            (0.105, 0.146359)
            (0.12, 0.229885)
            (0.135, 0.333333)
            (0.15, 0.435256)
        };

    \addplot[
        color = kit-green,
        line width = 0.75pt,
        mark options = {solid},
        mark size = 2pt,
        mark = square,
        dashed,
        dash pattern=on 4pt off 2pt,
    ]
        coordinates {
            (0.015, 3e-06)
            (0.03, 9.4e-05)
            (0.045, 0.000802)
            (0.06, 0.00424)
            (0.075, 0.01461)
            (0.09, 0.03675)
            (0.105, 0.07817)
            (0.12, 0.14238)
            (0.135, 0.23288)
            (0.15, 0.33991)
        };

    \addplot[
        color = kit-blue,
        line width = 0.75pt,
        mark options = {solid},
        mark size = 1.5pt,
        mark= *
    ]
        coordinates {
            (0.015, 2e-06)
            (0.03, 5.7e-05)
            (0.045, 0.000522)
            (0.06, 0.00275)
            (0.075, 0.010074)
            (0.09, 0.027075)
            (0.105, 0.057811)
            (0.12, 0.111869)
            (0.135, 0.192348)
            (0.15, 0.294942)
        };

    \addplot[
        color = kit-purple, %
        line width = 0.75pt,
        mark options = {solid},
        mark size = 2pt,
        mark= triangle*
    ]
        coordinates {
            (0.015, 1.3e-06)
            (0.03, 4.2e-05)
            (0.045, 0.000396)
            (0.06, 0.002082)
            (0.075, 0.007908)
            (0.09, 0.020344)
            (0.105, 0.047201)
            (0.12, 0.098449)
            (0.135, 0.165805)
            (0.15, 0.262526)
        };

  \end{axis}
\end{tikzpicture}
    \vspace{-1.15cm}
    \caption{LER of OBP4-aSCED$K$, \mbox{$K\in\{64,256\}$}, and state-of-the-art decoders, i.e., MWPM and correlated MWPM, on the $\llbracket 128,2,8\rrbracket$ toric code. 
    }
    \label{fig:toric_d8}
\end{figure}

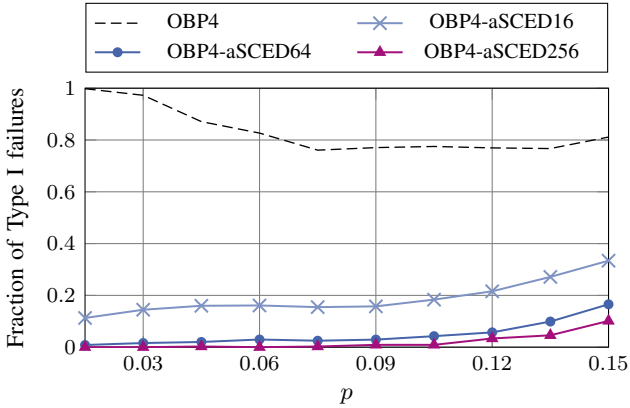
\begin{figure}[t]
    \centering
     
\definecolor{kit-green}{RGB}{0, 150, 130}
\definecolor{kit-blue}{RGB}{70, 100, 170}
\definecolor{kit-red}{RGB}{162, 34, 35}
\definecolor{kit-orange}{RGB}{223, 135, 27}
\definecolor{kit-yellow}{RGB}{252, 229, 0}
\definecolor{kit-purple}{RGB}{163, 16, 124}

\definecolor{aSCED256}{RGB}{245, 0, 131}
 
\begin{tikzpicture}
    \pgfplotsset{grid style={gray}}
	\pgfplotsset{every tick label/.append style={font=\footnotesize}}
  \begin{axis}[
    xshift = 1.5cm,
    xmin = 0.015,   
    xmax = 0.15,
    ymin = 0,   
    ymax = 1,
    xtick = {0.03, 0.06, 0.09, 0.12, 0.15},
    xticklabel style = {/pgf/number format/fixed},
    axis background/.style = {fill=white, mark size=1.5pt},
    xmajorgrids,
    ymajorgrids,
    width = 8.5cm,
    height = 5.cm,
    xlabel = {$p$},
    ylabel = {Fraction of Type I failures},
    label style={font=\small},
    legend cell align={left},
    legend style={anchor = south, at={(0.5,1.06)}, fill opacity=1, text opacity = 1, legend columns=2, font=\footnotesize, row sep = 0pt, /tikz/every even column/.append style={column sep=0.25cm}, text width = 2.45cm}
    ]
 
    \addlegendentry{OBP4\phantom{-aSCED64}}
    \addlegendentry{OBP4-aSCED16}
    \addlegendentry{OBP4-aSCED64}
    \addlegendentry{OBP4-aSCED256}
 
    \addplot[
        color = black,
        line width = 0.5pt,
        mark options = {solid},
        mark size = 2pt,
        dashed,
        dash pattern=on 4pt off 2pt,
    ]
        coordinates {
            (0.015, 0.998077)
            (0.030, 0.972659)
            (0.045, 0.871000)
            (0.060, 0.826719)
            (0.075, 0.760836)
            (0.090, 0.770677)
            (0.105, 0.775046)
            (0.120, 0.769534)
            (0.135, 0.766970)
            (0.150, 0.811404)
        };
 
    \addplot[
        color = kit-blue!70,
        line width = 0.75pt,
        mark options = {solid},
        mark size = 3.5pt,
        mark= x,
    ]
        coordinates {
            (0.015, 0.112826)
            (0.030, 0.144538)
            (0.045, 0.159925)
            (0.060, 0.161052)
            (0.075, 0.154706)
            (0.090, 0.157593)
            (0.105, 0.183685)
            (0.120, 0.215801)
            (0.135, 0.271345)
            (0.150, 0.334015)
        };
 
    \addplot[
        color = kit-blue,
        line width = 0.75pt,
        mark options = {solid},
        mark size = 1.5pt,
        mark= *
    ]
        coordinates {
          (0.015, 0.00851)
          (0.030, 0.01599)
          (0.045, 0.02013)
          (0.060, 0.02961)
          (0.075, 0.02495)
          (0.090, 0.02931)
          (0.105, 0.04254)
          (0.120, 0.05721)
          (0.135, 0.09900)
          (0.150, 0.16554)
        };
 
    \addplot[
        color = kit-purple, %
        line width = 0.75pt,
        mark options = {solid},
        mark size = 2pt,
        mark= triangle*
    ]
        coordinates {
            (0.015, 0.000000)
            (0.030, 0.000000)
            (0.045, 0.002695)
            (0.060, 0.000000)
            (0.075, 0.002564)
            (0.090, 0.009063)
            (0.105, 0.008969)
            (0.120, 0.033613)
            (0.135, 0.045872)
            (0.150, 0.101391)
        };
 
    \end{axis}
\end{tikzpicture}
    \vspace{-1.15cm}
    \caption{Fraction of Type I failures relative to the total number of logical errors of OBP4 and OBP4-aSCED$K$, \mbox{$K\in\{16,64,256\}$}, on the $\llbracket 128,2,8\rrbracket$ toric code, as a metric for the convergence reliability of the decoders.}
    \label{fig:toric_d8_convergence}
\end{figure}

Having discussed the parametrization of aSCED, we now compare aSCED using overcomplete matrices against state-of-the-art decoders, namely binary minimum weight perfect matching (MWPM)~\cite{MWPM2002}, correlated MWPM (cMWPM)~\cite{fowler2013optimalcomplexitycorrectioncorrelated}, and (binary) BP with OSD post-processing~\cite{Panteleev2021} (BP-OSD).
For the numerical results of binary MWPM and cMWPM, we use the Pymatching library \cite{Pymatching}. 
For binary MWPM, $X$-~and $Z$ errors are decoded independently by separate MWPM instances operating on $\bm H_Z$ and $\bm H_X$, respectively. The final error estimate of binary MWPM is obtained by combining the binary estimates of both instances.
The numerical results of BP-OSD are obtained similarly.

Fig.~\ref{fig:toric_d8} shows the LER of OBP4-aSCED$K$, binary MWPM, and cMWPM on the $\llbracket 128,2,8\rrbracket$ toric code. 
We observe that OBP4-aSCED$K$ outperforms both MWPM decoders for larger ensemble sizes, i.e., $K\in\{64,256\}$.

Fig.~\ref{fig:toric_d8_convergence} shows the fraction of Type I failures relative to the total number of logical errors for OBP4 and \mbox{OBP4-aSCED$K$} on $\llbracket 128,2,8\rrbracket$ toric code.
Note that we did not include stand-alone BP4 since we observed that almost all its failures are of Type I.
OBP4 reduces the fraction of Type I failures to $0.75$ for \mbox{$p\!=\!0.075$}.
\mbox{OBP4-aSCED$K$} further mitigates the non-convergence issues of OBP4, reducing this fraction to $0.15$ for \mbox{$K\!=\!16$} and $0.0025$ for \mbox{$K\!=\!256$}.
Further simulations show that for the $\llbracket 46,2,9\rrbracket$ GB code, \mbox{OBP4-aSCED64} exhibits no Type I failures across the entire simulated regime, i.e., always converges to a valid estimate. In practice, this largely alleviates the need for a post-processing step such as OSD.

\begin{figure}[t]
    \centering
    \definecolor{kit-green}{RGB}{0, 150, 130}
\definecolor{kit-blue}{RGB}{70, 100, 170}
\definecolor{kit-red}{RGB}{162, 34, 35}
\definecolor{kit-orange}{RGB}{223, 155, 27}
\definecolor{kit-yellow}{RGB}{252, 229, 0}
\definecolor{kit-purple}{RGB}{163, 16, 124}

\definecolor{aSCED256}{RGB}{245, 0, 131}

\begin{tikzpicture}
    \pgfplotsset{grid style={gray}}
	\pgfplotsset{every tick label/.append style={font=\footnotesize}}
  \begin{axis}[
    xshift = 1.5cm,
    xmin = 0.03,   
    xmax = 0.1,
    ymin = 1e-4,   
    ymax = 1,
    ymode = log,
    xtick = {0.04, 0.06, 0.08, 0.1},
    xticklabel style = {/pgf/number format/fixed},
    axis background/.style = {fill=white, mark size=1.5pt},
    xmajorgrids,
    ymajorgrids,
    width = 8.5cm,
    height = \ferfigheight,
    xlabel={$p$},
    ylabel={LER},
    label style={font=\small},
    legend cell align={left},
    legend style={anchor = south east,  at={(1,0)}, fill opacity=1, text opacity = 1,legend columns=1,font=\footnotesize, row sep = 0pt}
  ]

    \addlegendentry{BP-OSD}
    \addlegendentry{OBP4}
    \addlegendentry{OBP4-aSCED16}
    \addlegendentry{OBP4-aSCED64}
    \addlegendentry{OBP4-aSCED256}

    \addplot[
        color = kit-orange,
        line width = 0.75pt,
        mark options = {solid},
        mark size = 2pt,
        mark = o,
        dashed,
        dash pattern=on 4pt off 2pt,
    ]
        coordinates {
            (0.03, 0.012354)
            (0.04, 0.036704)
            (0.05, 0.082491)
            (0.06, 0.159236)
            (0.07, 0.267917)
            (0.08, 0.373832)
            (0.09, 0.539084)
            (0.1, 0.651466)
        };

    \addplot[
        color = black,
        line width = 0.75pt,
        mark options = {solid},
        mark size = 2pt,
        dashed,
        dash pattern=on 4pt off 2pt,
    ]
        coordinates {
            (0.02, 0.000256)
            (0.03, 0.002132)
            (0.04, 0.008963)
            (0.05, 0.03032)
            (0.06, 0.078154)
            (0.07, 0.157155)
            (0.08, 0.257311)
            (0.09, 0.378696)
            (0.1, 0.540577)
        };

    \addplot[
      color = kit-blue!70,
        line width = 0.75pt,
        mark options = {solid},
        mark size = 3pt,
        mark= x,
    ]
        coordinates {
            (0.03, 0.000313)
            (0.04, 0.002172)
            (0.05, 0.009174)
            (0.06, 0.028651)
            (0.07, 0.068217)
            (0.08, 0.139971)
            (0.09, 0.26854)
            (0.1, 0.390999)
        };

    \addplot[
        color = kit-blue,
        line width = 0.75pt,
        mark options = {solid},
        mark size = 1.5pt,
        mark= *
    ]
        coordinates {
            (0.03, 0.000194)
            (0.04, 0.001332)
            (0.05, 0.006132)
            (0.06, 0.02068)
            (0.07, 0.05869)
            (0.08, 0.117591)
            (0.09, 0.217914)
            (0.1, 0.335894)
        };

    \addplot[
        color = kit-purple, %
        line width = 0.75pt,
        mark options = {solid},
        mark size = 2pt,
        mark= triangle*
    ]
        coordinates {
            (0.03, 0.000141)
            (0.04, 0.001038)
            (0.05, 0.004455)
            (0.06, 0.016611)
            (0.07, 0.046351)
            (0.08, 0.098692)
            (0.09, 0.187484)
            (0.1, 0.306691)
        };

  \end{axis}
\end{tikzpicture}
    \vspace{-1.15cm}
    \caption{LER of OBP4-aSCED$K$, \mbox{$K\in\{16,64,256\}$}, and BP-OSD on the $\llbracket 126,28,8\rrbracket$ GB code. BP and OBP4 are limited to \mbox{$I_\text{max}\!=\!200$} decoding iterations. OSD post-processing is applied at order $10$.}
    \label{fig:GB_OC}
\end{figure}

Finally, Fig.~\ref{fig:GB_OC} shows the LER of OBP4-aSCED$K$ and BP-OSD on the $\llbracket 126,28,8 \rrbracket$ GB code, i.e., for a code with \mbox{$k>2$}. 
Note that OSD post-processing is applied at order $10$.
For OBP4-aSCED$K$, each batch uses 
\mbox{$\Delta\!=\!2$} splitters of weight~$6$. Furthermore, we limit both BP-OSD and OBP4 to \mbox{$I_\mathrm{max}\!=\!200$} iterations, and fix $p_0\!=\!0.1$ for the latter. Again, OBP4-aSCED$K$ yields strong decoding performance and outperforms BP-OSD significantly.

\section{Conclusion and Outlook}
In this work, we characterize degeneracy by introducing \emph{degeneracy sets}. 
We prove that appending linearly independent rows (\emph{splitters}) to the check matrix partitions each degeneracy set into disjoint subsets of equal size.
Each constituent path of aSCED then exploits this splitting effect by decoding under the influence of a smaller degeneracy set, while the whole ensemble still preserves full degeneracy-awareness.
Simulation results demonstrate that aSCED significantly improves convergence compared to stand-alone BP and outperforms correlated MWPM for the $\llbracket 128,2,8\rrbracket$ toric code.
In particular, for the $\llbracket 46,2,9\rrbracket$ GB code, aSCED can fully mitigate Type~I failures due to non-convergence and reduces the amount of Type~II failures (logical error rate).
Future work will extend our framework beyond the depolarizing channel to incorporate fault-tolerant decoding scenarios.

\newpage

\IEEEtriggeratref{23}


\begin{thebibliography}{10}
\providecommand{\url}[1]{#1}
\csname url@samestyle\endcsname
\providecommand{\newblock}{\relax}
\providecommand{\bibinfo}[2]{#2}
\providecommand{\BIBentrySTDinterwordspacing}{\spaceskip=0pt\relax}
\providecommand{\BIBentryALTinterwordstretchfactor}{4}
\providecommand{\BIBentryALTinterwordspacing}{\spaceskip=\fontdimen2\font plus
\BIBentryALTinterwordstretchfactor\fontdimen3\font minus \fontdimen4\font\relax}
\providecommand{\BIBforeignlanguage}[2]{{%
\expandafter\ifx\csname l@#1\endcsname\relax
\typeout{** WARNING: IEEEtran.bst: No hyphenation pattern has been}%
\typeout{** loaded for the language `#1'. Using the pattern for}%
\typeout{** the default language instead.}%
\else
\language=\csname l@#1\endcsname
\fi
#2}}
\providecommand{\BIBdecl}{\relax}
\BIBdecl

\bibitem{Shor1995}
P.~W. Shor, ``Scheme for reducing decoherence in quantum computer memory,'' \emph{Phys. Rev. A}, vol.~52, pp. R2493--R2496, 1995.

\bibitem{Gottesman1997}
D.~Gottesman, ``Stabilizer codes and quantum error correction,'' Ph.D. dissertation, California Institute of Technology, 1997.

\bibitem{wilde2013quantum}
M.~M. Wilde, \emph{Quantum Information Theory}.\hskip 1em plus 0.5em minus 0.4em\relax Cambridge university press, 2013.

\bibitem{MacKay2004}
D.~J.~C. MacKay, G.~Mitchison, and P.~L. McFadden, ``Sparse-graph codes for quantum error correction,'' \emph{IEEE Trans. Inf. Theory}, vol.~50, no.~10, pp. 2315--2330, 2004.

\bibitem{TillichZemor2014}
J.-P. Tillich and G.~Z{\'e}mor, ``Quantum {LDPC} codes with positive rate and minimum distance proportional to the square root of the blocklength,'' \emph{IEEE Trans. Inf. Theory}, vol.~60, no.~2, pp. 1193--1202, Feb. 2014.

\bibitem{Panteleev2021}
P.~Panteleev and G.~Kalachev, ``Degenerate quantum {LDPC} codes with good finite length performance,'' \emph{Quantum}, vol.~5, p. 585, Nov. 2021.

\bibitem{bravyi_high-threshold_2024}
S.~Bravyi, A.~W. Cross, J.~M. Gambetta, D.~Maslov, P.~Rall, and T.~J. Yoder, ``High-threshold and low-overhead fault-tolerant quantum memory,'' \emph{Nature}, vol. 627, no. 8005, pp. 778--782, Mar. 2024.

\bibitem{Gallager1962}
R.~G. Gallager, ``Low-density parity-check codes,'' \emph{IRE Trans. Inf. Theory}, vol.~8, no.~1, pp. 21--28, 1962.

\bibitem{shor1996quantumerrorcorrectingcodesneed}
P.~W. Shor and J.~A. Smolin, ``Quantum error-correcting codes need not completely reveal the error syndrome,'' in \emph{Proc. Annu. Symp. Found. Comput. Sci. (FOCS)}, Burlington, VT, USA, Oct. 1996.

\bibitem{iyer2015hardness}
P.~Iyer and D.~Poulin, ``Hardness of decoding quantum stabilizer codes,'' \emph{IEEE Trans. Inf. Theory}, vol.~61, no.~9, pp. 5209--5223, 2015.

\bibitem{PoulinChung2008}
D.~Poulin and Y.~Chung, ``On the iterative decoding of sparse quantum codes,'' \emph{Quantum Inf. Comput.}, vol.~8, no.~10, pp. 987--1000, Nov. 2008.

\bibitem{GuidedDecimation2024}
H.~Yao, W.~{Abu Laban}, C.~H{\"a}ger, A.~{Graell i Amat}, and H.~D. Pfister, ``Belief propagation decoding of quantum {LDPC} codes with guided decimation,'' in \emph{Proc. IEEE Int. Symp. Inf. Theory (ISIT)}, Athens, Greece, Jul. 2024.

\bibitem{MSLS25}
S.~Miao, A.~Schnerring, H.~Li, and L.~Schmalen, ``Quaternary neural belief propagation decoding of quantum {LDPC} codes with overcomplete check matrices,'' \emph{IEEE Access}, vol.~13, pp. 25\,637--25\,649, Feb. 2025.

\bibitem{Raveendran2021}
N.~Raveendran and B.~Vasi\'{c}, ``Trapping sets of quantum {LDPC} codes,'' \emph{Quantum}, vol.~5, p. 562, Oct. 2021.

\bibitem{mueller2025improvedbeliefpropagationsufficient}
T.~M{\"u}ller, T.~Alexander, M.~E. Beverland, M.~Bühler, B.~R. Johnson, T.~Maurer, and D.~Vandeth, ``Improved belief propagation is sufficient for real-time decoding of quantum memory,'' 2025.

\bibitem{roffe_decoding_2020}
J.~Roffe, D.~R. White, S.~Burton, and E.~T. Campbell, ``Decoding across the quantum {LDPC} code landscape,'' \emph{Phys. Rev. Res.}, vol.~2, no.~4, p. 043423, Dec. 2020.

\bibitem{hillmann_localized_2025}
T.~Hillmann, L.~Berent, A.~O. Quintavalle, J.~Eisert, R.~Wille, and J.~Roffe, ``\BIBforeignlanguage{en}{Localized statistics decoding for quantum low-density parity-check codes},'' \emph{\BIBforeignlanguage{en}{Nat. Commun.}}, vol.~16, no.~1, p. 8214, Sep. 2025.

\bibitem{yin_symbreak_2024}
K.~Yin, X.~Fang, J.~Ruan, H.~Zhang, D.~Tullsen, A.~Sornborger, C.~Liu, A.~Li, T.~Humble, and Y.~Ding, ``{SymBreak}: Mitigating quantum degeneracy issues in {QLDPC} code decoders by breaking symmetry,'' Dec. 2024, 10.48550/arXiv.2412.02885.

\bibitem{LK21}
C.-Y. Lai and K.-Y. Kuo, ``Log-domain decoding of quantum {LDPC} codes over binary finite fields,'' \emph{IEEE Trans. Quantum Eng.}, vol.~2, pp. 1--15, Sep. 2021.

\bibitem{liu_neural_2019}
Y.-H. Liu and D.~Poulin, ``\BIBforeignlanguage{en}{Neural belief-propagation decoders for quantum error-correcting codes},'' \emph{\BIBforeignlanguage{en}{Phys. Rev. Lett.}}, vol. 122, no.~20, p. 200501, May 2019.

\bibitem{Geiselhart2022}
M.~Geiselhart, M.~Ebada, A.~Elkelesh, J.~Clausius, and S.~ten Brink, ``Automorphism ensemble decoding of quasi-cyclic {LDPC} codes by breaking graph symmetries,'' \emph{IEEE Commun. Lett.}, vol.~26, no.~8, pp. 1705--1709, Aug. 2022.

\bibitem{mandelbaum2024endomorphisms}
J.~Mandelbaum, S.~Miao, H.~J{\"a}kel, and L.~Schmalen, ``Endomorphisms of linear block codes,'' in \emph{Proc. IEEE Int. Symp. Inform. Theory (ISIT)}, Athens, Greece, Jul. 2024.

\bibitem{Hehn2010}
T.~Hehn, J.~Huber, O.~Milenkovic, and S.~Laendner, ``Multiple-bases belief-propagation decoding of high-density cyclic codes,'' \emph{IEEE Trans. Commun.}, vol.~58, no.~1, pp. 1--8, Jan. 2010.

\bibitem{MMJS24}
S.~Miao, J.~Mandelbaum, H.~Jäkel, and L.~Schmalen, ``A joint code and belief propagation decoder design for quantum {LDPC} codes,'' in \emph{Proc. IEEE Int. Symp. Inform. Theory (ISIT)}, Athens, Greece, Jul. 2024.

\bibitem{MJS25}
J.~Mandelbaum, H.~Jäkel, and L.~Schmalen, ``Subcode ensemble decoding of linear block codes,'' in \emph{Proc. IEEE Int. Symp. Inform. Theory (ISIT)}, Ann Arbor, MI, USA, Jun. 2025.

\bibitem{AED_RMcodes}
M.~Geiselhart, A.~Elkelesh, M.~Ebada, S.~Cammerer, and S.~ten Brink, ``Automorphism ensemble decoding of {Reed-Muller} codes,'' \emph{IEEE Trans. Commun.}, vol.~69, no.~10, pp. 6424--6438, Oct. 2021.

\bibitem{koutsioumpas_automorphism_2025}
S.~Koutsioumpas, H.~Sayginel, M.~Webster, and D.~E. Browne, ``Automorphism ensemble decoding of quantum {LDPC} codes,'' Mar. 2025, 10.48550/arXiv.2503.01738.

\bibitem{koutsioumpas_colour_2025}
S.~Koutsioumpas, T.~Noszko, H.~Sayginel, M.~Webster, and J.~Roffe, ``Colour codes reach surface code performance using {Vibe} decoding,'' Aug. 2025, 10.48550/arXiv.2508.15743.

\bibitem{maan2026decodingcorrelatederrorsquantum}
A.~S. Maan, F.~M. Garcia~Herrero, A.~Paler, and V.~Savin, ``Decoding correlated errors in quantum {LDPC} codes,'' \emph{Nat. Commun.}, vol.~17, p. 3965, Mar. 2026.

\bibitem{gong2024lowlatencyiterativedecodingqldpc}
A.~Gong, S.~Cammerer, and J.~M. Renes, ``Toward low-latency iterative decoding of {QLDPC} codes under circuit-level noise,'' Mar. 2024, 10.48550/arXiv.2403.18901.

\bibitem{ye2025beamsearchdecoderquantum}
\BIBentryALTinterwordspacing
M.~Ye, D.~Wecker, and N.~Delfosse, ``Beam search decoder for quantum {LDPC} codes,'' 2025. [Online]. Available: \url{https://arxiv.org/abs/2512.07057}
\BIBentrySTDinterwordspacing

\bibitem{MBtB26}
J.~Mandelbaum, P.~Bezner, S.~ten Brink, H.~Jäkel, and L.~Schmalen, ``Affine subcode ensemble decoding for linear block codes,'' Apr. 2026, 10.48550/arXiv.2604.06889.

\bibitem{MWPM2002}
E.~Dennis, A.~Kitaev, A.~Landahl, and J.~Preskill, ``Topological quantum memory,'' \emph{J. Math. Phys.}, vol.~43, no.~9, pp. 4452--4505, Sep. 2002.

\bibitem{Kitaev2003}
A.~Y. Kitaev, ``Fault-tolerant quantum computation by anyons,'' \emph{Ann. Phys.}, vol. 303, no.~1, pp. 2--30, Jan. 2003.

\bibitem{Calderbank1998}
A.~R. Calderbank, E.~M. Rains, P.~W. Shor, and N.~J.~A. Sloane, ``Quantum error correction via codes over {GF(4)},'' \emph{{IEEE} Trans. Inf. Theory}, vol.~44, no.~4, pp. 1369--1387, Jul. 1998.

\bibitem{Calderbank1996}
A.~R. Calderbank and P.~W. Shor, ``Good quantum error-correcting codes exist,'' \emph{Phys. Rev. A}, vol.~54, no.~2, pp. 1098--1105, Aug. 1996.

\bibitem{Steane1996}
A.~M. Steane, ``Multiple-particle interference and quantum error correction,'' \emph{Proc. R. Soc. Lond. A}, vol. 452, no. 1954, pp. 2551--2577, Sep. 1996.

\bibitem{Roffe2019}
J.~Roffe, ``Quantum error correction: an introductory guide,'' \emph{Contemp. Phys.}, vol.~60, no.~3, p. 226–245, Jul. 2019.

\bibitem{LK20}
K.-Y. Kuo and C.-Y. Lai, ``Refined belief propagation decoding of sparse-graph quantum codes,'' \emph{IEEE J. Sel. Areas in Inf. Theory}, vol.~1, no.~2, pp. 487--498, Aug. 2020.

\bibitem{Overcomplete2006}
S.~Laendner, T.~Hehn, O.~Milenkovic, and J.~B. Huber, ``When does one redundant parity-check equation matter?'' in \emph{Proc. Global Conf. Commun. (GLOBECOM)}, San Francisco, CA, USA, Nov. 2006.

\bibitem{Leon1988}
J.~S. Leon, ``A probabilistic algorithm for computing minimum weights of large error-correcting codes,'' \emph{IEEE Trans. Inf. Theory}, vol.~34, no.~5, pp. 1354--1359, Sep. 1988.

\bibitem{Pymatching}
O.~Higgott and C.~Gidney, ``Sparse blossom: Correcting a million errors per core second with minimum-weight matching,'' \emph{Quantum}, vol.~9, p. 1600, Jan. 2025.

\bibitem{fowler2013optimalcomplexitycorrectioncorrelated}
A.~G. Fowler, ``Optimal complexity correction of correlated errors in the surface code,'' 2013, 10.48550/arXiv.1310.0863.

\end{thebibliography}
\end{document}